\documentclass{wp}

\title{A simple proof of the representation theorem \\
for betweenness preferences}
\author{Yutaro Akita}
\address{Dept.~of Economics, Penn State University}
\email{ytrakita@gmail.com}
\date{June 26, 2024}
\thanks{I am grateful to Kalyan Chatterjee
for helpful suggestions and discussions.}
\jelcodes{D80, D81}

\begin{document}

\begin{abstract}
  This paper presents a simple proof
  of \citet{Dekel1986}'s representation theorem for betweenness preferences.
  The proof is based on the separation theorem.
\end{abstract}

\maketitle

\section{Introduction}

In the theory of non-expected utility under risk,
\citet{Dekel1986} characterizes the class of preferences
satisfying the betweenness axiom.
He shows that betweenness preferences
have implicit expected utility representations.
\citet{Payro2023} generalizes the representation theorem
to preferences on a compact convex subset of a vector space and
applies it to study temptation and self-control problems.

Their proofs, however, are long and involved.
This paper presents a shorter and more intuitive proof.
\citet{Conlon1995} also provides an alternative proof,
but I use a different technique.
My proof is based on the separation theorem:
the upper and lower contour sets of the preference are separated
by a hyperplane, which identifies the local utility index.
To deal with the case where the domain has infinite dimensions,
I use a technique similar to that of \citet{ChatterjeeKrishna2008}:
I first apply the separation theorem to construct the local utility index
on a finite-dimensional subset and then extend it to the whole domain.

\section{Model}

Let $X$ be a nonempty compact convex subset
of a Hausdorff topological vector space.
I take as primitive a binary relation $\wpr$ on $X$.
This formulation is general enough to cover
preferences over lotteries on a compact metrizable space
\citep[studied in][]{Dekel1986}, and
preferences over compact convex menus of lotteries \citep{DLR2001,Payro2023}.
As usual,
denote by $\ipr$ and $\spr$ the symmetric and asymmetric parts of $\wpr$,
respectively.

The following axioms are standard.

\begin{axiom}[Rationality]\label{axm:rat}
  The relation $\wpr$ is complete and transitive.
\end{axiom}

\begin{axiom}[Nondegeneracy]\label{axm:nond}
  There exists $(x, y) \in X^2$ such that $x \spr y$.
\end{axiom}

\begin{axiom}[Continuity]\label{axm:cont}
  The upper and lower contour sets of $\wpr$ are closed everywhere.
\end{axiom}

The key axiom is the following,
which is studied by \citet{Dekel1986} and \citet{Chew1989}.

\begin{axiom}[Betweenness]\label{axm:bet}
  For each $(\lambda, (x, y)) \in (0, 1) \times X^2$,
  if $x \spr y$, then $x \spr \lambda x + (1 - \lambda)y \spr y$.
\end{axiom}

Together with \cref{axm:rat,axm:cont},
the betweenness axiom \ref{axm:bet} implies neutrality toward mixing.

\begin{lemma}\label{lem:bet}
  If $\wpr$ satisfies \cref{axm:rat,axm:cont,axm:bet},
  then $x \ipr y$ implies $x \ipr \lambda x + (1 - \lambda)y \ipr y$
  for each $(\lambda, (x, y)) \in (0, 1) \times X^2$.
\end{lemma}

\begin{proof}
  Assume \cref{axm:cont,axm:bet}.
  Choose any $(\lambda, (x, y)) \in (0, 1) \times X^2$ with $x \ipr y$.
  Suppose that we can take $z \in X$ be such that $y \spr z$.
  Otherwise, we can use the similar argument
  by taking $z \in X$ with $z \spr x$.
  Define the sequence $\seq{w_n}$ by $w_n = (1 - n^{-1})y + n^{-1}z$,
  which converges to $y$.
  For each $n \in \SN$,
  since $x \ipr y \spr w_n$ by \cref{axm:bet},
  we have $x \spr \lambda x + (1 - \lambda)w_n \spr w_n$
  again by \cref{axm:bet}.
  Thus, from \cref{axm:rat,axm:cont},
  $x \wpr \lambda x + (1 - \lambda)y \wpr y$;
  that is, $x \ipr \lambda x + (1 - \lambda)y$.
\end{proof}

A real-valued function $f$ on $X$
is \emph{mixture linear}
if $f(\lambda x + (1 - \lambda)y) = \lambda f(x) + (1 - \lambda)f(y)$
for each $(\lambda, (x, y)) \in [0, 1] \times X^2$;
a real-valued function $U$ on $X$ \emph{represents $\wpr$}
if $x \wpr y$ is equivalent to $U(x) \ge U(y)$.
An \emph{implicit mixture linear representation of $\wpr$}
is a real-valued function $u$ on $X \times [0, 1]$ such that
\begin{enumerate}
  \item
    $u(\cdot, t)$ is continuous and mixture linear for each $t \in (0, 1)$;
  \item
    $u(x, \cdot)$ is continuous on $(0, 1)$ for each $x \in X$;
  \item
    there exists $(x^*, x_*) \in X^2$
    such that $u(x^*, t) = 1$ and $u(x_*, t) = 0$ for each $t \in [0, 1]$;
  \item
    for each $x \in X$,
    there exists unique $t \in [0, 1]$ such that $t = u(x, t)$;
  \item
    the continuous real-valued function $U$ on $X$
    of the form $U(x) = u(x, U(x))$ represents $\wpr$.
\end{enumerate}

In the rest of the paper, I prove the following theorem.

\begin{theorem}[\citealp{Dekel1986,Payro2023}]\label{prop:bet}
  The relation $\wpr$ satisfies \cref{axm:rat,axm:nond,axm:cont,axm:bet}
  if and only if it admits an implicit mixture linear representation.
\end{theorem}

\section{Proof of the theorem}

The necessity of the axioms is routine.
We only show the sufficiency.
Assume \cref{axm:rat,axm:nond,axm:cont,axm:bet}.

By \cref{axm:rat,axm:cont},
there exists $(x^*, x_*) \in X^2$
such that $x^* \wpr x \wpr x_*$ for each $x \in X$.
By \cref{axm:nond}, $x^* \spr x_*$.
For each $t \in [0, 1]$, let $m_t = t x^* + (1 - t)x_*$.
For each $(x, t) \in X \times (0, 1)$,
let $\xi_t(x) = x_*$ if $x \wpr m_t$ and let $\xi_t(x) = x^*$ otherwise.
Then, by \cref{axm:rat,axm:cont,axm:bet},
there exist unique functions $U \colon X \to [0, 1]$ and
$\mu \colon X \times (0, 1) \to (0, 1]$
such that $x \ipr m_{U(x)}$ for each $x \in X$ and
$m_t \ipr \mu(x, t) x + (1 - \mu(x, t))\xi_t(x)$,
which are continuous by \cref{axm:cont}.
By \cref{axm:bet}, $U$ represents $\wpr$.

Let $\cP$ be the collection
of all polytopic subsets of $X$ that include $\{x^*, x_*\}$.
For each $(t, P) \in (0, 1) \times \cP$,
since by \cref{axm:cont,axm:bet} and \cref{lem:bet},
$\set{x \in P \mvert x \wpr m_t}$ and
$\set{x \in P \mvert m_t \spr x}$ are nonempty convex sets
having disjoint relative interiors,
it follows
from the separation theorem \citep[Theorem 11.3]{Rockafellar1970} that
there exists a nonconstant affine functional $v_t^P$ on the affine hull of $P$
such that for each $x \in P$,
\begin{equation}
  x \wpr m_t \iff v_t^P(x) \ge v_t^P(m_t), \qquad
  x \spr m_t \iff v_t^P(x) > v_t^P(m_t).
\end{equation}
Here, passing to a normalization,
we may assume $v_t^P(x^*) = 1$ and $v_t^P(x_*) = 0$.
Then, for each $(x, t) \in X \times (0, 1)$,
the value $v_t^P(x)$ is independent of the choice of $P$ as long as $x \in P$:
for each $P \in \cP$ with $x \in P$,
we have $t = v_t^P(m_t)
= \mu(x, t) v_t^P(x) + (1 - \mu(x, t))v_t^P(\xi_t(x))$, so
\begin{equation}\label{eq:v_expl}
  v_t^P(x)
  = \frac{t - (1 - \mu(x, t))v_t^P(\xi_t(x))}{\mu(x, t)}
  =
  \begin{dcases*}
    \frac{t}{\mu(x, t)}         & if $U(x) \ge t$, \\
    1 - \frac{1 - t}{\mu(x, t)} & otherwise.
  \end{dcases*}
\end{equation}

For each $x \in X$,
let $C(x)$ be the convex hull of $\{x^*, x_*, x\}$.
Define $u \colon X \times [0, 1] \to \SR$ by
\begin{equation}
  u(x, t) =
  \begin{dcases*}
    \vone_{X \setminus I(x_*)}(x) & if $t = 0$, \\
    v_t^{C(x)}(x)                 & if $t \in (0, 1)$, \\
    \vone_{I(x^*)}(x)             & if $t = 1$,
  \end{dcases*}
\end{equation}
where $I(z) = \set{x \in X \mvert x \ipr z}$ for each $z \in \{x^*, x_*\}$ and
$\vone_A$ is the indicator function of $A$ on $X$.
We show that $u$ is an implicit mixture linear representation of $\wpr$.
For each $t \in (0, 1)$,
the mixture linearity of $u(\cdot, t)$ is inherited from $v^P_t$'s.
From \eqref{eq:v_expl} and the continuity of $\mu$,
it follows that $u$ is continuous on $X \times (0, 1)$.
By construction, $u(x^*, t) = 1$ and $u(x_*, t) = 0$ for each $t \in [0, 1]$.
Finally,
$1 = u(x, 1)$ if and only if $x \ipr x^*$;
$0 = u(x, 0)$ if and only if $x \ipr x_*$;
for each $t \in (0, 1)$, we have $t = u(x, t)$ if and only if
$t = v_t^{C(x)}(x)$,
which is equivalent to $t = U(x)$.
\qed

\begin{remark}
  If $X$ is finite-dimensional,
  it is not necessary to restrict the upper and lower contour sets
  to a polytopic subset of $X$:
  we can just use the separation theorem to find a separating functional.
  This approach fails, however, if $X$ has infinite dimensions.
  For example,
  suppose $X$ is the set of all Borel probability measures on $[0, 1]$,
  endowed with the weak$*$ topology.
  Then, $X$ has empty relative interior,
  which makes the separation theorem inapplicable.
  By focusing on a polytopic (i.e., finite-dimensional) subset of $X$,
  I overcame this difficulty.
\end{remark}

\bibliographystyle{refs}
\bibliography{refs}

\end{document}